\documentclass[a4paper, 11pt]{amsart}
\newtheorem{theorem}{Theorem}
\newtheorem{lemma}{Lemma}

\theoremstyle{remark}
\newtheorem{remark}{Remark}

\usepackage{amsmath, amssymb, amsthm}
\usepackage{tikz}
\usetikzlibrary{arrows}

\newcommand{\NN}{\mathbb{N}}

\newcommand{\abs}[1]{\vert #1 \vert}

\newcommand{\uu}{{\mathbf u}}
\newcommand{\ww}{{\mathbf w}}
\newcommand{\sv}{{\mathbf v}}

\newcommand{\lex}{\triangleleft}

\begin{document}
\title{Words with unbounded periodicity complexity}
\author{\v St\v ep\' an Holub}
\address{Department of Algebra, Charles University, Sokolovsk\'a 83, 175 86 Praha, Czech Republic}
\email{holub@karlin.mff.cuni.cz}
\thanks{Supported by the Czech Science Foundation grant number 13-01832S}

\begin{abstract}
If an infinite non-periodic word is uniformly recurrent or is of bounded repetition, then the limit of its periodicity complexity is infinity. Moreover, there are uniformly recurrent words with the periodicity complexity arbitrarily high at infinitely many positions. 
 \end{abstract}
\maketitle

\section{Introduction}
In \cite{antonio}, a new complexity function of infinite words, called \emph{periodicity complexity}, was introduced. It  gives, for any position in the word, the average value of the local periods up to that position. The authors construct an infinite word for which the periodicity complexity is bounded. Since the word is not uniformly recurrent, they ask whether there exist non-periodic uniformly recurrent words having bounded periodicity complexity (see Remark 3.8. in \cite{antonio}). In Section \ref{alphasec}, we define words with special lexicographic properties. In Section  \ref{factorsec} we use those words to define a factorization of any non-periodic uniformly recurrent word. This factorization allows to show, in Section \ref{recsec},  that the answer to the above question is negative. Moreover, in Section \ref{repsec}, we show that any word with bounded repetition has unbounded periodicity complexity too.

The authors of \cite{antonio} also prove that the periodicity complexity can exceed any fixed function infinitely many times. Again, the witness word is not uniformly recurrent, and a problem is left open whether the same is true for uniformly recurrent words (see Remark 3.8. in \cite{antonio}). In Section \ref{function} we construct a uniformly recurrent word showing that the answer is positive. The method can be seen as a generalized Toeplitz construction.

\section{Basic concepts}
We first recall basic definitions and concepts. Let $w=a_1a_2a_3\cdots$, where $a_i$ are letters, be a finite or infinite word over some alphabet $\Sigma$. A \emph{position} in the word is any integer $1\leq i\leq |w|$ (we have $|w|=\infty$ if $w$ is infinite). The position $i$ can be understood as the border between $a_{i}$ and $a_{i+1}$, but we will rather identify it with the pair $(u,v)$ of words such that $w=uv$ and $|u|=i$. Note that, unlike \cite{antonio}, we consider also the position $|w|$ for a finite word, which does not lie between two letters and corresponds to $(w,\varepsilon)$, where $\varepsilon$ denotes the empty word.

By \emph{proper prefix} of $w$ we mean any prefix, including the empty one, that is strictly shorter than $w$.
The \emph{period} of $w$, denoted by $p(w)$, is the least integer $p$ such that $a_i=a_{i+p}$ holds for all $i\geq 1$ satisfying $i+p\leq |w|$. A finite word $w$ is called \emph{unbordered} if $p(w)=|w|$. Otherwise, $w$ is called \emph{bordered}, and any proper nonempty prefix of $w$ that is also a suffix of $w$ is called its \emph{border}. It is important and easy to see that the shortest border of a bordered word is itself unbordered.

If 
\[\sup\{e : \text{$v^e$ is a factor of $u$ for some nonempty word $v$}\}\]
for some infinite word $u$ is finite, we say that $u$ is of \emph{bounded repetition}.

A word is \emph{primitive} if it is not a power of a shorter word. A primitive word $w$ is a \emph{Lyndon word} if it is lexicographically minimal within the set of all words conjugate with it. That is, $w\lex vu$ for any factorization $w=uv$, where $\lex$ is a lexicographic order. It is well known that Lyndon words are unbordered.

A \emph{repetition word} at the position $(u,v)$ is any nonempty word $r$ that is suffix comparable with $u$ and prefix comparable with $v$ (any word being prefix comparable with the empty word). 
The \emph{local period} of $w$, denoted by $p_w(i)$, is defined by
\[p_w(i)=\min\{|r| : \text{$r$ is a repetition word at the position $i$}\}.\]
Note that an infinite word may have positions $(u,v)$ with no repetition word. This happen if and only if there is no repetition word of length at most $|u|$ and $u$ is not a factor of $v$. Then the corresponding local period is $\infty$ (in accordance with the usual definition of $\min \emptyset$). Note also that $p_w(|w|)=1$ for any finite $w$.

The most important result concerning local periods is the Critical Factorization Theorem, which states that for any finite word there is a position $i$ such that $p_w(i)=p(w)$.  Such a position is called \emph{critical}. For a proof of the Critical Factorization Tehorem using lexicographic orderings, see \cite{cft}.

Let $w=uvz$ be a factorization of $w$. We will say that the position $1\leq i\leq |v|$ of $v$ \emph{corresponds} to the position $|u|+i$ of $w$. We also say that the position $i$ of $w$ lies in $v$ if $|u|<i\leq |uv|$. This is rather informal since a particular occurrence of $v$ in $w$ is implicitly understood, but it will be always clear from the context.

Note that 
\begin{align}\label{factor}
	p_v(i)\leq p_{uvz}(|u|+i), \quad 1\leq i \leq |v|.
\end{align}
 Informally, the local period of a word at a position is at least the local period of its factor at the corresponding position.

We say that $j\geq 0$ is an \emph{occurrence} of a word $z$ in $w$ if $w=uv$, $|u|=j$ and $z$ is a prefix of $v$. If $j$ and $j'$ are two consecutive occurrences of $z$ in $w=a_1a_2a_3\cdots$, then $a_{j+1}a_{j+2}\cdots a_{j'}$ is called a \emph{return word to $z$} in $w$. Return words were introduced and studied in \cite{durand}.

An infinite word $w$ is called \emph{recurrent} if any factor of $w$ has an infinite number of occurrences in $w$. It is called \emph{uniformly recurrent} if for any factor $z$, the length of all return words to $z$ in $w$ is bounded. We call the length of the longest return word to $z$ in $w$ the \emph{maximal return time} of $z$ in $w$. Note that if an infinite word is recurrent, then its local period is finite at all positions.

The property studied in this paper is the \emph{periodicity complexity} of $w$ which is the function $h_w$ defined on positions of $w$ by
\[h_w(i)=\frac 1i\sum_{j=1}^ip_w(j).\]
The notion was introduced (for infinite words) in \cite{antonio}. If $w$ is infinite, then it is reasonable to suppose that $w$ is recurrent, which guarantees that the range of $h_w$ are positive rationals (no infinity has to be dealt with). However, we can also put $h_w(i)=\infty$ if $p_w(j)=\infty$ for some $1\leq j\leq i$.
For a finite $w$, denote \[h(w):=h_w(\abs{w}).\] Note the following useful inequality, which follows from \eqref{factor}:
\begin{align}\label{hcka}
	h(uv)\geq \frac{|u|\cdot h(u)+ |v|\cdot h(v)}{|uv|}.
\end{align}

\section{Lexicographically minimal return words}\label{alphasec}
Let $\uu$ be an infinite non-periodic uniformly recurrent word. Such a word has the following property.
\begin{lemma}\label{exp}
Let $v$ be a factor of $\uu$. Then the set of integers $e$ such that $v^e$ is a factor of $\uu$ is finite.
\end{lemma}
\begin{proof}
Since $\uu$ is non-periodic, it contains a word $w$ with the period greater than $|v|$. Let $m$ be the maximal return time of $w$ in $\uu$. Then any word of length $m+|w|$ contains $w$ as a factor. This implies that $v^e$, which has the period smaller than $w$, has to be shorter than $m+|w|$.  
\end{proof}

 We define recursively an infinite sequence of words $\alpha_k$, $k\geq 1$. Fix a lexicographic order $\lex$ with the least letter $a$, and let $\alpha_1=a$. For $k> 1$, let $e_{k}$ be the largest integer such that $\alpha_{k}^{e_{k}}$ is a factor of $\uu$. The definition of $e_k$ is correct by Lemma \ref{exp}.  Now $\alpha_{k+1}$ is defined as the lexicographically minimal return word to $\alpha_{k}^{e_{k}}$ in $\uu$. 

\begin{lemma}\label{unbordered}
For each $k\geq 1$, the word $\alpha_k$ is the lexicographically minimal factor of $\uu$ of length $\abs{\alpha_k}$ and
each return word to $\alpha_k^{e_k}$ is Lyndon. In particular, the word $\alpha_k$ is unbordered for each $k\geq 1$.
\end{lemma}
\begin{proof}
Proceed by induction. Clearly, $\alpha_1=a$ is lexicographically minimal word of length one and unbordered. Let $k>1$. Since $\alpha_{k-1}$ is unbordered and $e_k$ is maximal, two distinct occurrences of $\alpha_{k-1}^{e_{k-1}}$ in $\uu$ do not overlap. Therefore $\alpha_{k-1}^{e_{k-1}}$ is a prefix of $\alpha_k$. From lexicographic minimality of $\alpha_{k-1}$, we deduce that $\alpha_{k-1}^{e_{k-1}}$ is a prefix of the lexicographically minimal factor of $\uu$ of length $\abs{\alpha_k}$. Such a word is therefore prefix comparable with some return word to $\alpha_{k-1}^{e_{k-1}}$. The definition of $\alpha_k$ now implies that $\alpha_k$ is the lexicographically minimal factor of $\uu$ of its length.

Let $w$ be a return word to $\alpha_k^{e_k}$, $k\geq 1$. We first show that $w$ is unbordered. 
Suppose that $r$ is the shortest border of $w$ and  let $j$ be the largest integer such that $\alpha_j$ is a prefix of $r$. Clearly, $1\leq j< k$. The maximality of $e_j$ implies that $r\neq \alpha_j$, since $r\alpha_k$, and therefore also $r\alpha_j^{e_j}$, is a factor of $\uu$.  Recall that $r$ is unbordered since it is a shortest border. Therefore $\alpha_j^{e_j}$ is a proper prefix of $r$ (otherwise $r$ is bordered), and $r$ is a proper prefix of $\alpha_{j+1}$ (otherwise $j$ is not maximal). This implies that $r$ is a return word to $\alpha_j^{e_j}$ that is lexicographically smaller than $\alpha_{j+1}$, a contradiction.

Let now $w=uv$ be such that $vu$ is the Lyndon conjugate of $w$ and suppose that both $u$ and $v$ are nonempty. Since $w$ is a return word, the word $\alpha_{k}^{e_k}$ does not occur in $v$. The lexicographic minimality of $\alpha_k$ and $vu\lex w$ implies that $v$ is a prefix of $\alpha_k^{e_k}$. Therefore $w$ is bordered, a contradiction.

The word $\alpha_k$, $k>1$, is unbordered since it is a return word to $\alpha_{k-1}^{e_{k-1}}$.
\end{proof}

\section{Unbordered factorizations}\label{factorsec}
For any factor, an infinite recurrent word admits a factorization defined by occurrences of that factor. Such a factorization was considered already in \cite{durand}. We will study factorizations of $\uu$ given by $\alpha_k^{e_k}$. 
For each $k\geq 1$, let
 \[\uu=w_{k,0}w_{k,1}w_{k,2}w_{k,3}\cdots,\]
where $w_{k,0}\alpha_k^{e_k}$ is the shortest prefix of $\uu$ containing $\alpha_k^{e_k}$, and $w_{k,j}$ is the $j$th return word to $\alpha_k^{e_k}$ in $\uu$, for $j\geq 1$. In other words, the integer $\abs{w_{k,0}w_{k,1}\cdots w_{k,j-1}}$ is the $j$th occurrence of $\alpha_k^{e_k}$ in $\uu$.

 By Lemma \ref{unbordered}, all words $w_{k,j}$, $k\geq 0$, $j\geq 1$, are unbordered. Moreover, the $k$th factorization is a refinement of the $(k+1)$th one by the definition of $\alpha_k$. In particular, for each $k<k'$ and each $j\geq 1$, there are numbers $s,t\geq 1$ such that $w_{k',j}=w_{k,s}w_{i,s+1}\cdots w_{i,s+t}$.
We have already seen in the proof of Lemma \ref{unbordered} that two distinct occurrences of $\alpha_k^{e_k}$ in $\uu$ do not overlap. Therefore $\alpha_k^{e_k}$ is a prefix of $w_{k,j}$, for each $k,j\geq 1$.

Denote \[h_k:=\inf_{j\geq 1}\{h(w_{k,j})\}.\]
The following lemma is the core of the proof of Theorem \ref{recurrent}.
\begin{lemma}\label{h}
The sequence $(h_k)$ is unbounded.
\end{lemma}
\begin{proof}
We will show that for each $k$ there is some $k'>k$ such that $h_{k'}\geq h_k+\frac 12$.

Since $\uu$ is uniformly recurrent, the maximal return time of $\alpha_k^{m_k}$ in $\uu$ $$m_k:=\max_{j\geq 1} \{\abs{w_{k,j}}\}$$ is finite for each $k\geq 0$. On the other hand, the sequence $(\mu_k)$ where $$\mu_k:=\min_{j\geq 1} \{\abs{w_{k,j}}\}$$ is strictly growing since $\alpha_k^{e_k}$ is a prefix of each $w_{k,j}$ and  $|\alpha_k|$ is growing. Therefore, for each $k$, there is some $k'$ such that 
\begin{align}\label{minmax}
	\mu_{k'}>2m_k.
\end{align}
 We claim that $h_{k'}\geq h_k+\frac 12$ as required.

Chose $j\geq 1$ and let 
\[
w_{k',j}=w_{k,s}w_{k,s+1}\cdots w_{k,s+t}.
\]
In order to obtain a lower bound for $h(w_{k',j})$, estimate the local period at each position in $w_{k',j}$ by the local period at the corresponding position of a factor $w_{k,s'}$, $s\leq s'\leq s+t$, with only one exception: the critical position of $w_{k',j}$. At that position (we chose one of them if there are many) we shall insist on the actual value, which is $\abs{w_{k',j}}$ since $w_{k',j}$ is unbordered. 

Let $\ell$ be such that the chosen critical position of $w_{k',j}$ lies in $w_{k,\ell}$. Then the local period of $w_{k,\ell}$ at that position, which is at most $\abs{w_{k,\ell}}$, is replaced by $\abs{w_{k',j}}$.  
By \eqref{hcka}, we obtain the following bound.
\begin{align*}
	h(w_{k',j})&\geq \frac{\sum_{i=0}^{t} \abs{w_{k,s+i}}\cdot h(w_{k,s+i})+\abs{w_{k',j}}-\abs{w_{k,\ell}}}{\abs{w_{k',j}}} \\
	&\geq \frac{\sum_{i=0}^{t} \abs{w_{k,s+i}}\cdot h_k}{\abs{w_{k',j}}}+1-\frac {\abs{w_{k,\ell}}}{\abs{w_{k',j}}}\\
		&\geq h_k+1-\frac {\abs{w_{k,\ell}}}{\abs{w_{k',j}}}\geq h_k+\frac 12,
\end{align*}
where the last inequality follows from \eqref{minmax}. This completes the proof.
\end{proof}

\section{Uniformly recurrent words}\label{recsec}
We can now prove the first main result.
\begin{theorem}\label{recurrent}
Let $\uu$ be an infinite uniformly recurrent non-periodic word. Then
\[\lim_{i\to \infty} h_\uu(i)=\infty.\]   
\end{theorem}
\begin{proof}
For given $n$, we want to find $i_n$ such that, for each $i\geq i_n$, we have $h_\uu(i)\geq n$. 

Let $k$ be such that $h_k\geq n+1$. 
Denote 
\[
m=\max_{j\geq 0}\{\abs{w_{k,j}}\}=\max\{m_k,|w_{k,0}|\}\,.
\] 
Let $u$ be the prefix of $\uu$ of lenght $i\geq 2nm$. The word $u$ can be factorized as
\[u=w_{k,0}w_{k,1}w_{k,2}\dots w_{k,d}u',\]
where $u'$ is a proper prefix of $w_{k,d+1}$. The sum of local periods at positions of $u$ lying either in $w_{k,0}$ or in $u'$ is at least $\abs{w_{0,k}u'}$, and the sum of local periods for positions lying in $w_{k,1}w_{k,2}\dots w_{k,d}$ is at least
\[h_k\cdot \abs{w_{k,1}w_{k,2}\dots w_{k,d}}\geq (n+1)(\abs{u}-\abs{w_{k,0}u'}).\]
Therefore
\begin{align*}
	h_\uu(i)&\geq h(u)\geq \frac{\abs{w_{k,0}u'}+(n+1)(\abs{u}-\abs{w_{k,0}u'})}{\abs{u}}=\\
  					&=n+1-n\cdot \frac {\abs{w_{k,0}u'}}{\abs{u}}> n+1 -n\cdot \frac{2m}{2nm}=n.
\end{align*}
The last inequality uses $|u|\geq 2nm$ and $\abs{w_{k,0}u'}<2m$.
\end{proof}

\section{Bounded repetition}\label{repsec}
We shall now extend our result to words with bounded repetition. Let  $\uu$ be of bounded repetition and let $e<\infty$ be the largest integer such that $v^e$ is a factor of $\uu$ for some $v$. For each $k\geq 0$, we define factorizations
\[\uu=z_{k,0}z_{k,1}z_{k,2}\cdots,\] 
where $|z_{k,j}|=2^{k}$ for each $j\geq 0$. Then 
\[z_{k+i,j}=z_{k,2^{i}j}z_{k,2^{i}j+1}z_{k,2^{i}j+2}\cdots z_{k,2^{i}j+2^i-1}. \]
Denote
\[b_k:=\inf_{j\geq 1}\{h(z_{k,j})\}.\]
We can prove an analogue to Lemma \ref{h}.
\begin{lemma}\label{b}
The sequence $(b_k)$ is unbounded.
\end{lemma}
\begin{proof}
For given $k$, let $k'$ be such that 
\begin{align}\label{minmaxb}
	2^{k'}\geq e\cdot 2^{k+1}.
\end{align}
We show that $b_{k'}\geq b_k+\frac 14$.

Chose $j\geq 0$ and let $p$ be the period of $z_{k',j}$. Then 
\[\frac{2^{k'}}{p}\leq e,\]
whence
\begin{align}\label{p}
	p\geq \frac{2^{k'}}{e}\geq 2^{k+1}.
\end{align}
Let \[z_{k',j}=v^sv',\]
where $\abs{v}=p$ and $v'$ is a proper prefix of $v$. 
 
We bound the value $h(z_{k',j})$ from below similarly as in the proof of Lemma \ref{h}. Estimate the local period at each position of $z_{k',j}$ by the local period at the corresponding position of a factor $z_{k,j'}$, and then, for a chosen critical position of $v$ and each of $s$ occurrences of $v$, replace the previous value with $p$. Such a replacement increases the value by at least $p-2^k$ for each occurrence of $v$.
This yields the following bound.
\begin{align*}
	h(z_{k',j})&\geq b_k + \frac{s(p-2^k)}{2^{k'}}.
\end{align*}
Since $p$ is the period of $z_{k',j}$, it is easy to see that \[sp>\frac12 |z_{k',j}|.\]
From \eqref{p}, we deduce
\[\frac{s\cdot 2^k}{2^{k'}}\leq \frac 12 \cdot \frac{sp}{2^{k'}}.\]
Altogether, we have the desired inequality
\begin{align*}
	h(z_{k',j})&\geq b_k + \frac{s(p-2^k)}{2^{k'}}\geq b_k + \frac 12 \cdot \frac{sp}{2^{k'}}> b_k+\frac 14.
\end{align*}
This completes the proof.
\end{proof}
The following theorem is now easy to prove.
\begin{theorem}\label{repetition}
Let $\uu$ be an infinite word with bounded repetition. Then
\[\lim_{i\to \infty} h_\uu(i)=\infty.\]   
\end{theorem}
\begin{proof}
For a given $n$, let $k$ be such that $b_k\geq 2n$. Consider the prefix $u$ of $\uu$ of length $i\geq 2^{k}$. Then
\[u=z_{k',0}u',\]
where $k'\geq k$ and $u'$ is a proper prefix of $z_{k',1}$. We have
\begin{align*}
h_\uu(i)\geq h(u)\geq \frac{b_{k}\cdot |z_{k',0}|+|u'|}{|u|}\geq \frac{b_k}{2}\geq n.
\end{align*}
This completes the proof.
\end{proof}

\begin{remark}
	Theorem \ref{recurrent} and Theorem \ref{repetition} show that the construction of an infinite word with a bounded periodicity complexity as it is given in \cite{antonio} is the only possible in the following sense. Let $\uu$ be an infinite word that is not ultimately periodic and its periodicity complexity is bounded by $n$. Then $\uu$ can be factorized as
	\[\uu=v_0u_0^{e_0}v_1u_1^{e_1}v_2u_2^{e_2}\cdots,\] 
	where $h(u_i)<n$ for each $i\geq 0$, and the sequence of exponents $(e_i)$ is unbounded. 
\end{remark}

\section{High periodicity complexities}\label{function}
This section solves the problem left open in \cite{antonio}, Remark 3.23, by proving the following improvement of Theorem 3.20, ibidem. 
\begin{theorem}
For each function $f:\NN\to \NN$ there is a uniformly recurrent word $\uu$ such that $h_\uu(d)>f(d)$ for infinitely many integers $d$.
\end{theorem}

\begin{proof}
	Let $n_i$, $i\geq 1$, be a growing sequence of positive integers with $n_1\geq 2$. Then we define inductively a sequence $u_i$, $i\geq 0$, of words  by
	\begin{align*}
	u_0&=\varepsilon,\\
	u_{i}&=u_{i-1}a(u_{i-1}b)^{n_{i}}u_{i-1},\quad i\geq 1,
	\end{align*} 
	and put 
	\[\uu:= \lim_{i\to \infty}u_i.\]
We claim that 
 $\uu$ is uniformly recurrent, and for each $j\geq 1$, we have
\begin{align}\label{big}
p_\uu(d_j)=(n_j+1)(|u_{j-1}|+1),
\end{align}
 where 
\[
d_j=\abs{u_{j-1}au_{j-2}a\cdots u_1aa}.
\]
It is easy to verify inductively that 
\begin{align}\label{plus}
	u_ia, u_ib\in\{u_{i-1}a,u_{i-1}b\}^+.
\end{align}
This implies that, for each $i\geq 1$, the word $\uu$ can be factorized as a product of words $u_ia$ and $u_ib$. That is,
\[
\uu=u_ic_{1}u_ic_{2}u_ic_{3}\cdots,
\]
where $c_{j}$, $j\geq 1$, are letters. 
Therefore each factor $z$ whose first occurence lies within $u_i$ has the maximal return time bounded by $|u_i|+1$, and $\uu$ is uniformly recurrent.

We show that $\uu$ contains only occurrences of $u_i$, $i\geq 1$, visible in the above factorization. More precisely, if $k$ is an occurrence of $u_i$, then $k$ is a multiple of $|u_i|+1$. For $u_1=ab^{n_1}$, the property is easily verified. Proceed by induction.
Let $k$ be an occurrence  of $u_i$ in $\uu$ with $i\geq 2$.  The induction assumption and \eqref{plus} implies that 
\[k=k'\cdot (|u_i|+1)+\ell\cdot (|u_{i-1}|+1)\]
for some $k'\geq 0$ and $0\leq \ell\leq n_{i-1}+1$.   
The word $u_ic$, where  $c\in\{a,b\}$, contains at most two occurrences of $u_{i-1}a$, and only one of them is followed by $u_{i-1}b$, the one for which $\ell$ is zero. This proofs the claim since $u_{i-1}au_{i-1}b$ is a prefix of $u_i$.

We are ready to prove \eqref{big}. For each $j\geq 1$, denote 
\[s_j:=u_{j-1}au_{j-2}a\cdots u_1aa,\]
which is a prefix of $\uu$.
Let $r_j$ be the shortest repetition word at the position $d_j$. Observe that $d_1=1$ and $r_1=b^{n_1}a$, thus \eqref{big} holds for $j=1$.
We proceed by induction and consider two possibilities for $r_{j+1}$. \\
1. If $|r_{j+1}|<|s_j|+|u_{j-1}|+1$, then $r_{j+1}$ is a prefix of \[s_{j}^{-1}u_{j-1}a(u_{j-1}b)^2u_{j-1}.\] This implies that it is a repetition word at the position $d_j$ as well, which is a contradiction with the induction assumption.
\[
\begin{tikzpicture}
\def\lj{6} 
\draw[*-|] (0,0)--(4.5,0);\draw[-*] (4.5,0) -- (\lj,0);
\draw (10.5,0)-- (1.8*\lj,0);
\draw[dashed](1.8*\lj,0) -- (2*\lj,0);
\node at (0.5*\lj,-0.3){\small{$u_ja$}};
\node at (1.5*\lj,-0.3){\small{$u_jb\cdots$}};
\node at (5.25,0.2){\tiny{$u_{j-1}a$}};
\draw[-|] (\lj,0)--(7.5,0);
\node at (6.75,0.2){\tiny{$u_{j-1}a$}};
\node[draw,circle,inner sep=1pt] (d) at (8,0){};
\draw[-] (7.5,0)--(d);\draw[-|] (d)--(9,0);
\node at (8.25,0.2){\tiny{$u_{j-1}b$}};

\draw[-|] (9,0)--(10.5,0);
\node at (9.75,0.2){\tiny{$u_{j-1}b$}};
\node at (11.25,0.2){\tiny{$u_{j-1}\cdots$}};
%#############
\draw[rounded corners] (0.1,-0.3)|- (8,-0.8)--(8,-0.3);
\node at (4,-0.6) {\scriptsize{$s_{j+1}$}};
%#############
\draw[rounded corners] (5.9,-0.2)|- (8,-0.5)--(8,-0.2);
\node at (7,-0.3) {\scriptsize{$s_{j}$}};
%%%%
\draw (5,0.5) .. controls (5.5,1) and (7.5,1) .. node[fill=white]{\tiny{$r_{j+1}$}}  (8,0.5);
%%%%
\draw (8,0.5) .. controls (8.5,1) and (10.5,1) .. node[fill=white]{\tiny{$r_{j+1}$}} (11,0.5);
\end{tikzpicture}
\]

\noindent 2. Let now $|r_{j+1}|\geq |s_{j}|+|u_{j-1}|+1$. Then $(u_{j-1}a)^2$ is a factor of $r_{j+1}$. One can verify that the first occurrence of $(u_{j-1}a)^2$ greater or equal to $d_{j+1}$ is in the second occurrence of $s_{j+1}$, the latter being $(n_{j+1}+1)(|u_{j}|+1)$.
Therefore the word \[s_{j+1}^{-1}u_{j}a(u_{j}b)^{n_{j+1}}s_{j+1}\] of length predicted by \eqref{big} is the shortest repetition word at the position $d_{j+1}$.

To conclude the proof of the theorem, it is enough to define $n_j>2f(d_j)$ since then we have, for each $j\geq 1$,
\[
h_\uu(d_j)> \frac1{d_j}p_\uu(d_j)> f(d_j)\frac{2(|u_{j-1}|+1)}{d_j}>f(d_j).
\]
\end{proof}

We can also give an explicit formula for the $i$th letter  of the word $\uu$ constructed in the previous proof. Let $\uu=a_1a_2a_3\dots $, where $a_i\in \{a,b\}$, and let $m_0:=1$ and $m_j:=n_j+2$ for $j\geq 1$. Note that for each $j\geq 0$ we have $|u_j|+1=m_0m_1m_2\cdots m_j$. Then $a_i=a$ if and only if there is some $j\geq 0$ such that
\[
i\equiv m_0m_1m_2\cdots m_j \mod m_0m_1m_2\cdots m_j m_{j+1}. 
\]

Note that the word $\uu$ can be also obtained by the Toeplitz construction (see \cite{toeplitz,koskas}). Consider the alphabet $\{a,b,?\}$ and let $T(\ww,\sv)$, where $\ww$ and $\sv$ are infinite words over $\{a,b,?\}$, denote the infinite word obtained from $\ww$ by replacing
the sequence of all occurrences of  $?$ by $\sv$. Then we can define
\begin{align*}
\uu_0&=?^\omega, \\
\uu_{i}&=T(\uu_{i-1},(ab^{n_i}?)^\omega),\quad i\geq 1,
\end{align*} 
and \[\uu=\lim_{i\to \infty} \uu_i.\]

It is a task of further research to establish the level of control over the periodicity complexity function of resulting words given by the choice of the sequence  $(n_i)$ in this construction.

\bibliography{periodiccomplexity}{}

\begin{thebibliography}{1}

\bibitem{toeplitz}
Julien Cassaigne and Juhani Karhum{\"a}ki.
\newblock Toeplitz words, generalized periodicity and periodically iterated
  morphisms.
\newblock {\em Eur. J. Comb.}, 18(5):497--510, 1997.

\bibitem{cft}
Maxime Crochemore and Dominique Perrin.
\newblock Two-way string matching.
\newblock {\em J. ACM}, 38(3):651--675, 1991.

\bibitem{durand}
Fabien Durand.
\newblock A characterization of substitutive sequences using return words.
\newblock {\em Discrete Mathematics}, 179(1–3):89 -- 101, 1998.

\bibitem{koskas}
Michel Koskas.
\newblock Complexités de suites de toeplitz.
\newblock {\em Discrete Mathematics}, 183(1–3):161 -- 183, 1998.

\bibitem{antonio}
Filippo Mignosi and Anstonio Restivo.
\newblock A new complexity function for words based on periodicity.
\newblock {\em International Journal of Algebra and Computation},
  23(04):963--987, 2013.

\end{thebibliography}
\bibliographystyle{plain}

\end{document}